\documentclass[12pt]{article}

\usepackage[utf8]{inputenc}
\usepackage[T1]{fontenc}
\usepackage{amsmath,amssymb,amsfonts,amsthm}
\usepackage{mathtools}
\usepackage{bm}
\usepackage[margin=1.24in]{geometry}

\usepackage{booktabs}
\usepackage{graphicx}
\usepackage{url}
\usepackage{algorithm}
\usepackage{algpseudocode}
\usepackage{microtype}
\usepackage{xcolor}
\usepackage{subcaption}
\usepackage{float}
\usepackage{tikz}
\usetikzlibrary{arrows.meta, positioning}
\usepackage{tabularx}

\usepackage[colorlinks=true, citecolor=blue, linkcolor=blue, urlcolor=blue]{hyperref}
\usepackage[authoryear,round]{natbib}

\newtheorem{theorem}{Theorem}[section]

\newtheorem{proposition}[theorem]{Proposition}

\newtheorem{remark}[theorem]{Remark}

\newcommand{\myparagraph}[1]{\medskip\noindent\textbf{#1}}

\begin{document}

\title{An information-geometric framework for mapping maximum potential biodiversity}
\author{Shinto Eguchi\\[3mm]{The Institute of Statistical Mathematics, Japan}}
\date{\null}

\maketitle

\begin{abstract}
Biodiversity measures are often used descriptively: one computes a diversity
index from an observed or estimated community composition and maps the resulting
values across space.  Conservation planning, however, also requires a
site-specific benchmark against which the observed community can be compared.
This chapter develops an information-geometric framework for such
potential diversity and the associated diversity gap.  The central
object is a pair of probability vectors on the species simplex: an observed or
realized composition, and a potential composition obtained by a constrained variational principle.  The gap is
then defined by comparing a diversity functional at these two compositions.
The chapter takes Rao's quadratic entropy, which incorporates trait,
phylogenetic, or ecological dissimilarities among species, as its principal
diversity functional.  Hill-type diversity provides a complementary account
of species abundance and evenness.  A spatial point-process interpretation clarifies how local ecological
capacities can be defined before passing to the simplex.  Escort constraints,
capacity constraints, and divergence projections then provide a unified way to
define nontrivial benchmarks beyond the uniform distribution.  For Rao-type
potential diversity, we establish existence, sufficient conditions for
uniqueness, and a uniform perturbation bound for the optimal value and gap.
KL regularization supplies a unique selection of potential compositions,
with a quantitative bound on its effect on the Rao objective.
The resulting formulation separates two distinct questions: how diverse a
community is, and how far it is from a locally admissible potential benchmark.
It also connects the ecological idea of dark diversity with a continuous,
abundance-weighted comparison on the probability simplex.  We also outline a
dynamic extension in which capacities, species migration, and climate-driven
shifts vary over time.  Empirical implementation with large-scale
citizen-science biodiversity data and trait databases is left for future work.

\end{abstract}


\section{Introduction}

Biodiversity measures are usually descriptive.  Given a local community
composition \(p=(p_1,\ldots,p_S)\) on \(S\) species, one evaluates species
richness, Shannon entropy, Simpson diversity, a Hill number, or another
diversity index.  The resulting scalar is then interpreted as the diversity of
the community.  This descriptive use is indispensable, but it is not sufficient
for ecological planning.  A low-diversity site may be intrinsically constrained
by climate, habitat, or dispersal, whereas another site with the same observed
diversity may have a much larger ecological potential.  Conversely, a site with
moderate observed diversity may already be close to the best composition
allowed by its local constraints.

This chapter develops a framework for comparing \emph{realized diversity} with
\emph{potential diversity}.  The basic idea is to replace a single diversity map
by a pair of objects:
\[
p^{\mathrm{obs}}(x)
\quad\hbox{and}\quad
p^{\mathrm{pot}}(x),
\]
where \(p^{\mathrm{obs}}(x)\) denotes the observed or estimated community
composition at a location \(x\), and \(p^{\mathrm{pot}}(x)\) denotes a
benchmark composition obtained from an ecological constraint set.  Given a
diversity functional \( D \), the associated gap is
\[
\Delta_{ D }(x)
=
 D (p^{\mathrm{pot}}(x))
-
 D (p^{\mathrm{obs}}(x)).
\]
The purpose of this quantity is not to replace ordinary diversity indices, but
to add a benchmark-relative layer: it asks how much diversity is unrealized
under the same local admissible conditions.

The construction is information-geometric.  A potential composition is defined
by a constrained variational problem on the probability simplex, in the spirit
of maximum diversity distributions and their information-geometric
interpretation \citep{eguchi2025information}.  The primary target is Rao's
quadratic entropy \citep{Rao1982}, which incorporates pairwise dissimilarities
between species.  Rao-type gaps quantify losses of trait, phylogenetic, or
ecological spread that are invisible to abundance-only diversity measures.
The central mathematical distinction is between the optimal Rao diversity,
which is a uniquely defined scalar, and the optimizing composition, which
need not be unique.  We give conditions for uniqueness and stability, and a
KL-regularized selection principle when the optimum is not unique.

The complementary Hill-number family
\citep{Hill1973,Jost2006,ChaoChiuJost2014} unifies richness, Shannon
diversity, and Simpson diversity as effective numbers of species.
Hill-type gaps quantify abundance and evenness deficits.  Escort
distributions allow constraints to depend on the Hill order \(q\), so that
richness, evenness, and dominance receive different emphasis.

The chapter is theoretical.  No empirical case study is included.  This is a
deliberate choice.  Applying the framework to large citizen-science datasets,
such as spatial occurrence records, requires a separate treatment of sampling
bias, observation effort, temporal variation, and ecological interpretation.
Similarly, applying Rao-type gaps to functional diversity requires a careful
choice of trait databases and dissimilarity matrices.  These issues are
important, but they are better treated in a dedicated empirical study.  The
present chapter focuses on the mathematical foundation.

\myparagraph{Relation to dark diversity.}
The proposed gap is related to the ecological idea of dark diversity, namely
the set of species that are absent locally but belong to a regional species
pool \citep{partel2011dark}.  
The present formulation replaces a purely binary
presence--absence comparison by a benchmark-relative comparison on the
probability simplex.  For small values of the Hill order, the gap is close in
spirit to a richness deficit, whereas for larger values it reflects unrealized
evenness or dominance structure.  Thus potential diversity may be viewed as a
quantitative extension of dark diversity from species lists to
abundance-weighted community compositions.  This interpretation is only one
possible ecological reading of the gap; the mathematical definition remains the
variational contrast in \eqref{eq:gap-general}.

\subsection*{Organization}

Section~\ref{sec:functionals} begins with Rao's quadratic entropy and then
relates it to Simpson diversity, Hill numbers, and Shannon entropy.
Section~\ref{sec:potential-gap} defines potential diversity and diversity gaps
in a general form.  Section~\ref{sec:ig-construction} develops the
information-geometric construction using capacity constraints, escort
constraints, and divergence projection.  Section~\ref{sec:rao-potential}
develops Rao-type potential diversity, including its existence, uniqueness,
stability, and entropic selection.  Section~\ref{sec:hill-potential} then
develops the complementary Hill-type benchmark, its computation, and a
Shannon-type hierarchical decomposition.  Section~\ref{sec:dynamic} outlines a
dynamic extension for time-dependent capacities, migration, and warming-driven
community shifts.  Section~\ref{sec:conclusion} concludes with comments on
future empirical work.

\section{Diversity functionals}
\label{sec:functionals}

Let
\[
\Delta_{S-1}
=
\left\{
p=(p_1,\ldots,p_S): p_s\ge 0,\ \sum_{s=1}^S p_s=1
\right\}
\]
be the probability simplex over \(S\) species.  A diversity functional is a
mapping
\[
 D :\Delta_{S-1}\to \mathbb R
\]
that assigns a scalar measure of diversity to a community composition.  The
choice of \( D \) determines what aspect of biodiversity is emphasized.

\subsection{Rao's quadratic entropy}

Rao's quadratic entropy uses a dissimilarity matrix in addition to the
abundance vector \(p\):
\[
A=(a_{st})_{s,t=1}^S,
\qquad
a_{st}\ge 0,\quad a_{ss}=0,\quad a_{st}=a_{ts},
\]
where \(a_{st}\) measures the trait, phylogenetic, functional, or ecological
dissimilarity between species \(s\) and \(t\).  Rao's quadratic entropy is
\begin{equation*}
Q(p)
=
\sum_{s=1}^S\sum_{t=1}^S p_s p_t a_{st}
=
p^\top A p.
\end{equation*}
It is the expected dissimilarity between two species drawn independently from
the community composition \(p\).  If \(a_{st}=\|z_s-z_t\|^2\) for trait vectors
\(z_s\), then
\begin{equation*}
Q(p)
=
2\left\{
\sum_{s=1}^S p_s \|z_s\|^2
-
\left\|\sum_{s=1}^S p_s z_s\right\|^2
\right\}.
\end{equation*}
Thus Rao entropy is twice the trait variance of the community.

Rao entropy and Hill numbers answer different questions.  Hill numbers measure
how evenly probability mass is distributed across species labels.  Rao entropy
measures how widely the probability mass is spread in a dissimilarity space.
Two communities can have the same Hill diversity and different Rao diversity if
one community contains functionally similar species and the other contains
functionally distant species.

The link with Simpson diversity is especially direct.  If
\(a_{st}=1\) for \(s\ne t\) and \(a_{ss}=0\), then
\[
Q(p)=1-\sum_{s=1}^S p_s^2=1-D_2(p)^{-1}.
\]
Thus Rao's quadratic entropy generalizes the Gini--Simpson index by replacing
the binary distinction between different species with graded dissimilarities;
here \(D_2(p)=(\sum_s p_s^2)^{-1}\) is the inverse Simpson index discussed
below.

\subsection{Hill numbers}

For \(q\ge 0\), \(q\neq 1\), the Hill number of order \(q\) is
\begin{equation}
\label{eq:hill}
D_q(p)
=
\left(\sum_{s=1}^S p_s^q\right)^{1/(1-q)}.
\end{equation}
The continuous extension at \(q=1\) is
\begin{equation*}
D_1(p)
=
\exp\left\{-\sum_{s=1}^S p_s\log p_s\right\}.
\end{equation*}
Thus \(D_0\) is species richness, \(D_1\) is the exponential of Shannon entropy,
and \(D_2\) is the inverse Simpson concentration:
\[
D_0(p)=|\{s:p_s>0\}|,\qquad
D_2(p)=\left(\sum_{s=1}^S p_s^2\right)^{-1}.
\]
Hill numbers have the interpretation of effective numbers of species.  The
order \(q\) controls the sensitivity to rare and dominant species.  Small
values of \(q\) emphasize rare species and richness, while large values of \(q\)
emphasize dominant species.

It is often convenient to work with the monotone transform
\begin{equation*}
S_q(p)
=
\frac{1-\sum_{s=1}^S p_s^q}{q-1},
\qquad q\neq 1,
\end{equation*}
which is the Tsallis entropy.  For a fixed \(q\), maximizing \(D_q(p)\) is
equivalent to maximizing \(S_q(p)\), because \(D_q(p)\) is a monotone transform of
\(\sum_s p_s^q\) with the appropriate orientation.

\subsection{Shannon and Simpson diversity}

Shannon entropy is
\begin{equation*}
H(p)
=
-\sum_{s=1}^S p_s\log p_s.
\end{equation*}
It yields the effective diversity \(D_1(p)=\exp\{H(p)\}\).  Shannon entropy is
strictly concave on the interior of the simplex and plays a central role in
maximum entropy and information projection.

Simpson concentration is
\[
C_2(p)=\sum_{s=1}^S p_s^2,
\]
and Simpson diversity is often represented by
\[
1-C_2(p)
\quad\hbox{or}\quad
D_2(p)=C_2(p)^{-1}.
\]
The inverse form \(D_2\) has the same effective-number interpretation as the
other Hill numbers.  It is sensitive to dominance and therefore useful for
diagnosing whether a community is monopolized by a small number of abundant
species.

\section{Potential diversity and diversity gap}
\label{sec:potential-gap}

Let \(x\) denote a spatial location, habitat unit, or planning unit.  We assume
that an observed or realized community composition
\[
p^{\mathrm{obs}}(x)\in \Delta_{S-1}
\]
is available.  This vector may be obtained from direct counts, abundance
estimates, model-based predictions, or any other calibrated representation of
local composition.  The theory below does not depend on the particular
estimation method.

\subsection{A latent-intensity interpretation of local capacity}

Although the main framework is formulated on the probability simplex, it is
useful to give a probabilistic interpretation of the local capacity that enters
the construction.  Let \(\mathcal A\subset\mathbb R^d\) be a spatial domain and
let \(s=1,\ldots,S\) index the species.  We write \(\lambda_s(x)\) for the
latent ecological intensity of species \(s\) at location \(x\).  This intensity
represents the potential spatial occurrence rate before observational thinning,
detection error, and sampling effort are taken into account.  It is not a
species-composition probability.

A standard point-process interpretation is that the latent ecological process
of species \(s\) is an inhomogeneous Poisson point process with intensity
\(\lambda_s(x)\).  Thus, for any measurable region \(B\subset\mathcal A\),
\begin{equation*}
N_s(B)
\sim
\operatorname{Poisson}
\left\{
\int_B \lambda_s(u)\,du
\right\}.
\end{equation*}
This formulation is closely related to species distribution models and to the
well-known links among presence-only likelihoods, MaxEnt, Poisson point-process
models, and logistic regression approximations
\citep{ElithLeathwick2009,PhillipsAndersonSchapire2006,ElithEtAl2011,WartonShepherd2010,FithianHastie2013}.

Fix a local spatial scale \(\delta>0\), and write
\[
\mathcal N_\delta(x)
=
\{u\in\mathcal A:\|u-x\|\le \delta\}.
\]
This neighbourhood scale is used to define the observed local
composition.  Let \(N_s^{\mathrm{obs}}(x)\ge 0\) denote a local observed or
calibrated abundance score for species \(s\) around \(x\).  For example,
\(N_s^{\mathrm{obs}}(x)\) may be obtained from standardized counts, abundance
estimates, occupancy-adjusted predictions, or a calibrated species distribution
model.  When the total local abundance is positive, the observed composition is
defined by
\begin{equation*}
p_s^{\mathrm{obs}}(x)
=
\frac{N_s^{\mathrm{obs}}(x)}
{\sum_{k=1}^S N_k^{\mathrm{obs}}(x)} ,
\qquad s=1,\ldots,S .
\end{equation*}
Thus \(p^{\mathrm{obs}}(x)\) is a probability vector describing the realized
or calibrated community composition at the same spatial scale
\(\mathcal N_\delta(x)\).  It should not be confused with the latent capacity
scores introduced below.

The local ecological capacity of species \(s\) at \(x\) is defined as the
integrated latent intensity
\begin{equation}
\label{eq:local-capacity-lambda}
\Lambda_s(x)
=
\int_{\mathcal N_\delta(x)}
\lambda_s(u)\,du .
\end{equation}
Thus \(\Lambda_s(x)\) is a local capacity score, or a latent expected abundance
over the neighbourhood \(\mathcal N_\delta(x)\).  It is again not a probability.
The vector
\[
\Lambda(x)
=
(\Lambda_1(x),\ldots,\Lambda_S(x))
\]
summarizes the local ecological capacity of the species pool around \(x\).

The subsequent information-geometric construction uses \(\Lambda(x)\) only
through constraints on feasible community compositions.  In particular, the
main probability vectors in the framework are the observed composition
\(p^{\mathrm{obs}}(x)\) and the potential composition \(p^{\mathrm{pot}}(x)\).
The normalized vector proportional to \(\Lambda(x)\) may be useful as a
descriptive capacity profile, but it is not introduced as a separate benchmark
composition.  The potential composition \(p^{\mathrm{pot}}(x)\) is defined
later by a variational principle under capacity constraints, rather than by
direct normalization of \(\Lambda(x)\).

If an estimated intensity surface is available on a spatial grid, the local
capacity can be computed by a smoothing or aggregation operation.  For example,
with grid values \(\widehat L_s[j]\approx \widehat\lambda_s(u_j)\) and a kernel
\(K_\delta\) representing the neighbourhood scale, one may use
\[
\widehat\Lambda_s(x_j)
=
\sum_\ell K_\delta(x_j-u_\ell)\widehat L_s[\ell].
\]
For regular grids this is a convolution-type computation, while for irregular
sampling locations it is a local weighted average.  This numerical step is
separate from the variational definition of the potential composition.

\subsection{Sampling bias and calibrated composition}

The latent ecological intensity \(\lambda_s(x)\) is not observed directly in
empirical biodiversity data.  Detection, sampling effort, accessibility, and
survey design may thin the latent point process.  A compact representation is
\begin{equation*}
\lambda_s^{\mathrm{sam}}(x,e)
=
\rho_s(x,e)\lambda_s(x),
\qquad
0\le \rho_s(x,e)\le 1,
\end{equation*}
where \(e\) denotes effort or design information and \(\rho_s(x,e)\) is a
thinning probability.  Here \(\lambda_s^{\mathrm{sam}}(x,e)\) is the intensity
of the sampled or detected process, whereas \(\lambda_s(x)\) is the latent
ecological intensity.

For binary detection over a sufficiently small sampling unit, this thinning
interpretation leads to the complementary log-log form
\begin{equation*}
\log\{-\log(1-\Pr(Y_s=1\mid x,e))\}
=
\log \lambda_s(x)+\log \rho_s(x,e),
\end{equation*}
up to a sampling-area offset.  Models of this kind separate ecological intensity
from observation intensity, a distinction that is central in modern large-scale
biodiversity modeling \citep{KellingEtAl2019}.

This distinction also clarifies the meaning of the observed composition
\(p^{\mathrm{obs}}(x)\).  In the present framework, \(p^{\mathrm{obs}}(x)\) is
not the raw composition of detected records.  Rather, it is an observed,
estimated, or calibrated representation of the realized community composition
at location \(x\).  Depending on the data source, it may be obtained from direct
standardized counts, model-based abundance estimates, occupancy models,
point-process models with effort correction, or other calibrated species
distribution models.

If thinning is ignored, a diversity gap may confound ecological deficit with
survey deficit.  If thinning is modeled, the estimated latent intensity
\(\widehat\lambda_s(x)\) provides a route to the capacity vector
\(\widehat\Lambda(x)\) through \eqref{eq:local-capacity-lambda}.  The role of
the point-process model is therefore interpretive and constructive: it explains
how local capacity scores may be obtained, but the information-geometric
definition of \(p^{\mathrm{pot}}(x)\) does not require a particular observation
model.

The same idea is compatible with N-mixture models for imperfect count data.  In
a classical N-mixture formulation, a latent abundance \(N_i\) at site \(i\) is
modeled, for example, as
\[
N_i\sim \operatorname{Poisson}(\mu_i),
\]
and repeated observations are generated conditionally on \(N_i\) through a
detection model \citep{royle2004nmixture}.  If site \(i\) is identified with a
local neighbourhood \(\mathcal N_\delta(x)\), then the N-mixture mean
\(\mu_i\) corresponds to the integrated latent intensity
\[
\mu_i
=
\int_{\mathcal N_\delta(x)}
\lambda_s(u)\,du
=
\Lambda_s(x).
\]
Thus \(\Lambda_s(x)\) can be viewed as a continuous-space analogue of a latent
abundance mean in an N-mixture model.  This comparison is only a conceptual
bridge: the proposed diversity-gap framework does not require an N-mixture
model, nor does it identify the potential benchmark with a realized random
abundance.

\subsection{Potential compositions and gaps}

The central additional object is a potential composition
\[
p^{\mathrm{pot}}(x)\in \Delta_{S-1}.
\]
It is not intended to represent an observed community.  Rather, it is a
benchmark distribution that is admissible under local constraints and optimized
according to a specified principle.  The constraints may encode environmental
capacity, resource budgets, trait requirements, dispersal accessibility, or
management restrictions.

Let \(\mathcal C(x)\subset \Delta_{S-1}\) be a nonempty constraint set at location
\(x\).  For a diversity functional \( D \), a potential composition may
be defined by
\begin{equation}
\label{eq:pot-max-general}
p_{ D }^{\mathrm{pot}}(x)
\in
\underset{p(x)\in \mathcal C(x)}{\mathrm{argmax}}
 D (p(x)).
\end{equation}
The corresponding diversity gap is
\begin{equation}
\label{eq:gap-general}
\Delta_{ D }(x)
=
 D (p_{ D }^{\mathrm{pot}}(x))
-
 D (p^{\mathrm{obs}}(x)).
\end{equation}
When \(p^{\mathrm{obs}}(x)\in \mathcal C(x)\), this gap is nonnegative by
definition.  If \(p^{\mathrm{obs}}(x)\notin \mathcal C(x)\), then one may either
modify the constraint set so that the observed composition is admissible, or
interpret \(\Delta_{ D }\) as a benchmark contrast rather than a
mathematical deficit.

An alternative formulation uses a reference distribution \(r(x)\in\Delta_{S-1}\)
and an information divergence \(\mathcal D\).  The potential composition is then
defined by
\begin{equation*}
p^{\mathrm{pot}}(x)
=
\underset{p(x)\in \mathcal C(x)}{\mathrm{argmin}}
\mathcal D(p(x)\|r(x)).
\end{equation*}
This is an information projection.  In the Shannon case, \(\mathcal D\) is
typically the Kullback--Leibler divergence.  For power divergences, it may be an
\(\alpha\)-, \(\beta\)-, or \(\gamma\)-divergence \citep{cichocki2010families}.  The projection formulation
is useful when the goal is to find the closest feasible composition to a
reference state, rather than the most diverse feasible composition.

In this chapter we use the word \emph{potential} in a mathematical sense.  It
does not mean that the benchmark composition will be reached by natural
dynamics or management intervention.  It means that the composition is feasible
under the chosen constraints and optimal under the chosen variational
principle.  The ecological meaning of the potential benchmark depends on the
validity of the constraints.

\section{Information-geometric construction}
\label{sec:ig-construction}

This section describes the general information-geometric construction of
potential diversity.  The point-process discussion above explains one possible
probabilistic origin of a capacity vector, while the construction below treats
that vector as a mathematical input.  The ingredients are capacity constraints,
order-sensitive escort constraints (especially for Hill diversity), and
divergence projections relative to an explicit reference composition.
This viewpoint
is closely related to the information geometry of maximum diversity
distributions \citep{eguchi2025information}.

On the open simplex
\(\Delta_{S-1}^{\circ}=\{p\in\Delta_{S-1}:p_s>0\text{ for all }s\}\),
the tangent space is \(T_p\Delta_{S-1}^{\circ}
=\{u\in\mathbb R^S:\sum_su_s=0\}\).  The Fisher--Rao metric is
\begin{equation}
\label{eq:fisher-rao-metric}
g_p(u,v)=\sum_{s=1}^S\frac{u_sv_s}{p_s},
\qquad u,v\in T_p\Delta_{S-1}^{\circ}.
\end{equation}
The square-root embedding \(p\mapsto(2\sqrt{p_s})_{s=1}^S\)
identifies this metric with the metric induced on the positive orthant
of a sphere of radius two.  In particular, the Fisher--Rao geodesic
distance between interior compositions \(p\) and \(r\) is
\[
d_{\mathrm{FR}}(p,r)
=2\arccos\left(\sum_{s=1}^S\sqrt{p_sr_s}\right),
\]
and its minimizing geodesic is the corresponding great-circle arc in
square-root coordinates.  This intrinsic geometry is distinct from the
trait-dependent Rao quadratic form and from the Bregman projection
geometry used below; in general, a Bregman projection is not a
Fisher--Rao nearest-point projection.

The optimization here also differs from MaxEnt species-distribution
modeling: a spatial MaxEnt or Poisson point-process model describes the
geographic occurrence intensity of a species, whereas our benchmark
optimizes a probability vector \emph{across species} at a fixed location
under an already specified capacity vector.  The two constructions may
be used sequentially, with a fitted spatial model supplying capacities
for the subsequent diversity maximization.

\subsection{Capacity constraints}
\label{subsec:capacity}
Let
\[
\Lambda(x)=(\Lambda_1(x),\ldots,\Lambda_S(x))
\]
be the vector of nonnegative local capacity scores at location \(x\).
As introduced in equation~\eqref{eq:local-capacity-lambda}, \(\Lambda_s(x)\)
may be interpreted as an integrated latent ecological intensity over a local
neighbourhood of \(x\).  It is not a species-composition probability.
The value \(\Lambda_s(x)\) may represent habitat suitability, resource availability,
climatic admissibility, accessibility, or any other quantity that constrains
the feasible contribution of species \(s\).  In a purely theoretical treatment,
\(\Lambda_s(x)\) is simply a given covariate attached to species \(s\) at location
\(x\).

A basic linear capacity constraint is
\begin{equation*}
\sum_{s=1}^S p_s(x) \Lambda_s(x)
=
c(x),
\end{equation*}
where \(c(x)\) is a local capacity budget.  The feasible set is then
\[
\mathcal C(x)
=
\left\{
p(x)\in\Delta_{S-1}:
\sum_{s=1}^S p_s(x) \Lambda_s(x)=c(x)
\right\}.
\]
If \(c(x)\) is chosen as the observed average
\[
c(x)=\sum_{s=1}^S p_s^{\mathrm{obs}}(x)\Lambda_s(x),
\]
then \(p^{\mathrm{obs}}(x)\in\mathcal C(x)\), and the resulting potential gap is
nonnegative whenever the maximum in \eqref{eq:pot-max-general} exists.

\subsection{Escort constraints}

For \(q>0\), the \(q\)-escort distribution associated with \(p\) is
\begin{equation*}
p_s^{(q)}(p)
=
\frac{p_s^q}{\sum_{k=1}^S p_k^q}.
\end{equation*}
The escort expectation of a capacity vector is
\begin{equation*}
E^{(q)}[\Lambda;p]
=
\sum_{s=1}^S p_s^{(q)}(p)\Lambda_s.
\end{equation*}
The corresponding escort capacity constraint is
\begin{equation*}
E^{(q)}[\Lambda(x);p(x)]
=
c_q(x).
\end{equation*}
If
\[
c_q(x)=E^{(q)}[\Lambda(x);p^{\mathrm{obs}}(x)],
\]
then the observed composition is feasible.

Escort constraints are useful because they let the constraint depend on the
same order of emphasis as the diversity functional.  When \(q>1\), dominant
species receive more weight in the constraint.  When \(0<q<1\), rare species
receive relatively more weight.  This order-sensitive constraint is natural for
Hill-type diversity, because the order \(q\) already determines how the
diversity functional weights rare and common species.

\subsection{Divergence projection}
\label{sec:projection}
Let \(r(x)\in\Delta_{S-1}^{\circ}\) be a reference composition.  Given a strictly
convex function \(\Phi\) that is differentiable on the open simplex, the Bregman divergence
\[
\mathcal D_{\Phi}(p\|r)
=\Phi(p)-\Phi(r)- \langle\nabla \Phi(r), p-r\rangle
\]
defines a projection geometry on the simplex.  The potential composition can be
defined as
\[
p^{\mathrm{pot}}(x)
=
\underset{p(x)\in\mathcal C(x)}{\mathrm{argmin}}
\mathcal D_{\Phi}(p(x)\|r(x)).
\]
For the negative Shannon entropy
\[
\Phi(p)=\sum_{s=1}^S p_s\log p_s,
\]
the Bregman divergence is the Kullback--Leibler divergence.  The projection
then yields an exponential-family form under linear constraints.  For
power-type entropies, the associated projection yields \(q\)-exponential or
power-law forms.

The projection and maximum-diversity formulations are complementary.  Maximum
diversity asks for the most diverse feasible composition.  Projection asks for
the least distorted feasible composition relative to a reference.  In many
cases they coincide after a suitable choice of reference and entropy.

This complementarity can be made precise: the maximum-diversity formulation is
the special case of divergence projection in which the reference is uniform.
The associated diversity functional is defined as
\[
D_\Phi(p)=-\sum_{s=1}^S \phi(p_s).
\]
The Shannon case is \(\phi(t)=t\log t\); the Tsallis case is
\(\phi(t)=(t^{q}-t)/(q-1)\), whose monotone transform gives the Hill number
\(D_q\) in \eqref{eq:hill}.

\begin{proposition}[Maximum diversity as uniform-reference projection]
\label{prop:maxdiv-as-projection}
Let \(\Phi(p)=\sum_{s=1}^S \phi(p_s)\) be a separable Bregman generator and let
\(u=(1/S,\ldots,1/S)\) be the uniform composition.  Then, for every
\(p\in\Delta_{S-1}\),
\begin{equation}
\label{eq:maxdiv-as-projection}
\mathcal D_{\Phi}(p\|u)
=
-D_\Phi(p)-S\,\phi(1/S).
\end{equation}
Consequently, for any constraint set \(\mathcal C(x)\),
\[
\underset{p\in\mathcal C(x)}{\mathrm{argmin}}\,\mathcal D_{\Phi}(p\|u)
=
\underset{p\in\mathcal C(x)}{\mathrm{argmax}}\,\bigl(D_\Phi(p)\bigr),
\]
so projection onto \(\mathcal C(x)\) with a uniform reference coincides with
maximization of the separable diversity \(D_\Phi\) over \(\mathcal C(x)\).
\end{proposition}

\begin{proof}
By definition,
\(\mathcal D_{\Phi}(p\|u)=\Phi(p)-\Phi(u)-\langle\nabla\Phi(u),p-u\rangle\).
Because \(u\) is uniform, \(\nabla\Phi(u)_s=\phi'(1/S)\) is the same for every
\(s\); writing this common value as \(a\),
\[
\langle\nabla\Phi(u),p-u\rangle
=
a\sum_{s=1}^S (p_s-u_s)
=
a(1-1)=0,
\]
since \(p,u\in\Delta_{S-1}\).  Hence
\(\mathcal D_{\Phi}(p\|u)=\Phi(p)-\Phi(u)=-D_\Phi(p)-S\phi(1/S)\), which is
\eqref{eq:maxdiv-as-projection}.  The constant \(S\phi(1/S)\) does not affect
the argmin, giving the stated equivalence.
\end{proof}

The linear term vanishes only because the reference is uniform; for a
non-uniform reference \(\nabla\Phi(r)\) is no longer constant, the term
\(\langle\nabla\Phi(r),p\rangle\) survives, and \(\mathcal D_{\Phi}(p\|r)\) is no
longer an affine function of \(D_\Phi(p)\) alone.  Thus the maximum-diversity
benchmark is precisely the uniform-reference instance of the projection, and a
non-uniform reference is what allows the projection to express benchmarks—such
as the regional pool below—that maximum diversity cannot.

The projection formulation requires a reference composition \(r(x)\), and the
resulting benchmark depends on this choice.  The reference is best understood
as a declaration of which composition counts as having no unrealized
diversity, since the gap induced by the projection vanishes when
\(p^{\mathrm{obs}}(x)=r(x)\in\mathcal C(x)\).  Three choices are natural.

\begin{itemize}
\item \textbf{Regional pool composition.}  Taking \(r(x)\) to be the relative
abundance of the regional (or habitat-specific) species pool is the closest
analogue of dark diversity: the benchmark is then the feasible composition
nearest to the pool, and the gap measures the unrealized part of the pool
under the local capacity budget.
\item \textbf{Capacity profile.}  Taking \(r(x)\propto\Lambda(x)\) uses the
local capacity vector as a reference.  This is computationally convenient, but
\(\Lambda(x)\) is a latent intensity rather than a probability, so the
resulting benchmark should be read as a capacity-weighted, rather than an
ecological, reference.
\item \textbf{Uniform composition.}  Taking \(r(x)\) uniform recovers the
maximum-diversity benchmark, by Proposition~\ref{prop:maxdiv-as-projection}; in
the Shannon case this is the Gibbs form \eqref{eq:gibbs-pot} of
Section~\ref{subsec:maxent}.
\end{itemize}

The choice of \(r(x)\) is therefore not merely technical: it fixes the
ecological meaning of the benchmark, while the constraint set \(\mathcal C(x)\)
fixes what is locally admissible.  Separating the two—reference for the
target, constraint for the feasibility—is what makes the projection more
flexible than the maximum-diversity formulation.  As shown next, this
separation also has a precise geometric consequence.

\subsection{A Pythagorean decomposition}
\label{Pythagorean}

For an affine capacity constraint and an interior Bregman projection, the
divergence satisfies an exact Pythagorean relationship.  This includes
maximum-diversity projections with a uniform reference as a special case.

\begin{proposition}[Pythagorean relationship for the capacity projection]
\label{prop:pythagoras}
Let \(r(x)\in\Delta_{S-1}^{\circ}\) be a reference composition and let
\[
\mathcal C(x)
=
\left\{
p(x)\in\Delta_{S-1}:
\sum_{s=1}^S p_s(x)\Lambda_s(x)=c(x)
\right\}
\]
be the capacity set.  Let
\[
p^{\mathrm{pot}}(x)
\in
\underset{p(x)\in\mathcal C(x)}{\arg\min}\,
D_\Phi(p(x)\|r(x))
\]
be the Bregman projection of \(r(x)\) onto \(\mathcal C(x)\).  Assume that
\(p^{\mathrm{pot}}(x)\in\Delta_{S-1}^{\circ}\), so that no nonnegativity
constraint is active.  Then, for every
\(p\in\mathcal C(x)\),
\begin{equation}
\label{eq:pythagorean-capacity}
D_\Phi(p\|r(x))
=
D_\Phi(p\|p^{\mathrm{pot}}(x))
+
D_\Phi(p^{\mathrm{pot}}(x)\|r(x)).
\end{equation}
\end{proposition}

\begin{proof}
Fix the location \(x\) and write \(r=r(x)\), \(\Lambda=\Lambda(x)\),
\(c=c(x)\), and \(p^{\mathrm{pot}}=p^{\mathrm{pot}}(x)\).  The Bregman
divergence satisfies the three-point identity
\begin{equation}
\label{eq:three-point}
\mathcal D_{\Phi}(q\|r)
=
\mathcal D_{\Phi}(q\|p^{\mathrm{pot}})
+
\mathcal D_{\Phi}(p^{\mathrm{pot}}\|r)
+
\bigl\langle
\nabla\Phi(p^{\mathrm{pot}})-\nabla\Phi(r),\,
q-p^{\mathrm{pot}}
\bigr\rangle,
\end{equation}
which holds for all \(p^{\mathrm{pot}},q,r\) by direct expansion of the
definition of \(\mathcal D_{\Phi}\).  It therefore suffices to show that the
inner-product term vanishes for every \(q\in\mathcal C(x)\).

The feasible set \(\mathcal C(x)\) is the intersection of the two affine
constraints \(\sum_s p_s\Lambda_s=c\) and \(\sum_s p_s=1\).  Its set of
admissible directions is the linear subspace
\[
V
=
\Bigl\{v\in\mathbb R^S:\ \textstyle\sum_s v_s\Lambda_s=0,\ \sum_s v_s=0\Bigr\}
=
\mathrm{span}\{\Lambda,\mathbf 1\}^{\perp},
\]
where \(\mathbf 1=(1,\ldots,1)\).  Indeed, for any \(q,p^{\mathrm{pot}}\in
\mathcal C(x)\), the difference \(v=q-p^{\mathrm{pot}}\) satisfies
\(\sum_s v_s\Lambda_s=c-c=0\) and \(\sum_s v_s=1-1=0\), so \(v\in V\).

Since \(p^{\mathrm{pot}}\) minimizes \(f(p)=\mathcal D_{\Phi}(p\|r)\) over the
affine set \(\mathcal C(x)\), and \(\nabla_p f(p)=\nabla\Phi(p)-\nabla\Phi(r)\),
the first-order stationarity condition is
\begin{equation}
\label{eq:stationarity}
\bigl\langle
\nabla\Phi(p^{\mathrm{pot}})-\nabla\Phi(r),\,v
\bigr\rangle=0
\qquad\text{for all }v\in V.
\end{equation}
Equivalently, by the Lagrange conditions for the two affine constraints, there
exist multipliers \(\beta,\nu\in\mathbb R\) with
\(\nabla\Phi(p^{\mathrm{pot}})-\nabla\Phi(r)=\beta\Lambda+\nu\mathbf 1\), which
lies in \(\{\Lambda,\mathbf 1\}=V^{\perp}\); hence
\eqref{eq:stationarity} holds.

Applying \eqref{eq:stationarity} with \(v=q-p^{\mathrm{pot}}\in V\) shows that
the inner-product term in \eqref{eq:three-point} vanishes, which gives
\eqref{eq:pythagorean-capacity}.
\end{proof}
From Proposition \ref{prop:pythagoras}, if the observed composition satisfies
\(p^{\mathrm{obs}}(x)\in\mathcal C(x)\), then
\begin{equation}
\label{eq:pythagoras-obs}
\mathcal D_{\Phi}(p^{\mathrm{obs}}(x)\|r(x))
=
\mathcal D_{\Phi}(p^{\mathrm{obs}}(x)\|p^{\mathrm{pot}}(x))
+
\mathcal D_{\Phi}(p^{\mathrm{pot}}(x)\|r(x)).
\end{equation}
The decomposition \eqref{eq:pythagoras-obs} separates the total divergence of
the observed composition from the reference into a benchmark-relative deficit
and a structural term that reflects how far the capacity-constrained benchmark
itself lies from the reference.  The first term is the natural divergence
analogue of the diversity gap in \eqref{eq:gap-general}.  Both terms generally
depend on the reference \(r(x)\), since the projected composition itself
changes when the reference changes.

\section{Rao-type potential diversity}
\label{sec:rao-potential}

We now develop the principal benchmark of this chapter.  Rao-type gaps measure
deficits in functional, phylogenetic, or ecological dissimilarity spread;
the complementary Hill-type gaps of Section~\ref{sec:hill-potential} measure
abundance and evenness deficits.  Let \(A=(a_{st})\) be a species
dissimilarity matrix.  The observed Rao diversity is

\[
Q^{\mathrm{obs}}(x)
=
Q(p^{\mathrm{obs}}(x))
=
\sum_{s,t}p_s^{\mathrm{obs}}(x)p_t^{\mathrm{obs}}(x)a_{st}.
\]
A Rao-type potential composition is defined by
\begin{equation*}
p_R^{\mathrm{pot}}(x)
\in
\underset{p(x)\in\mathcal C_R(x)}{\mathrm{argmax}}
Q(p(x)),
\end{equation*}
where \(\mathcal C_R(x)\) is a capacity or trait-feasibility set.  The Rao
diversity gap is
\begin{equation*}
\Delta^R(x)
=
Q(p_R^{\mathrm{pot}}(x))
-
Q(p^{\mathrm{obs}}(x)).
\end{equation*}

The constraint set may be the same as in the Hill case, for example
\[
\mathcal C_R(x)
=
\left\{
p(x)\in\Delta_{S-1}:
\sum_s p_s(x) \Lambda_s(x)
=
\sum_s p_s^{\mathrm{obs}}(x)\Lambda_s(x)
\right\}.
\]
It may also include trait constraints such as
\[
\sum_s p_s(x) z_s=m(x),
\]
where \(z_s\) is a trait vector and \(m(x)\) is a local trait budget.  The
choice depends on the ecological meaning of the potential benchmark.

\subsection{Concavity and well-posedness}

The behavior of Rao maximization depends on the dissimilarity matrix.
Conditional negative semidefiniteness means
\[
v^\top A v\le0\qquad\text{whenever }\mathbf1^\top v=0.
\]
It holds, for example, for squared Euclidean dissimilarities
\(a_{st}=\|z_s-z_t\|^2\).  For the affine capacity set above, the weaker
condition on the feasible tangent space
\[
\mathcal T(x)=\{v\in\mathbb R^S:
\mathbf1^\top v=0,\ \Lambda(x)^\top v=0\}
\]
already suffices for concavity on \(\mathcal C_R(x)\).
Importantly, the maximal Rao \emph{value} is well defined even when its
maximizing \emph{composition} is not unique.

\begin{theorem}[Existence, uniqueness, and stability of Rao potential diversity]
\label{thm:rao-wellposed-stable}
Fix a location \(x\) and suppress \(x\) in the notation.  Let
\(\mathcal C=\{p\in\Delta_{S-1}:\Lambda^\top p=c\}\) be nonempty,
and let \(p^{\mathrm{obs}}\in\mathcal C\).  For a symmetric matrix \(A\),
define
\[
V_A=\max_{p\in\mathcal C}p^\top Ap,
\qquad
G_A=V_A-(p^{\mathrm{obs}})^\top Ap^{\mathrm{obs}}.
\]
The following statements hold.
\begin{enumerate}
\item The maximum \(V_A\) is attained and \(G_A\ge0\).
\item If \(v^\top Av\le0\) for all
\(v\in\mathcal T=\{v:\mathbf1^\top v=\Lambda^\top v=0\}\),
then \(Q_A\) is concave on \(\mathcal C\).  If the inequality is strict
for every nonzero \(v\in\mathcal T\), the maximizing composition is unique.
\item For every other symmetric matrix \(\widetilde A\), with the
\emph{same} feasible set and observed composition,
\begin{equation}
\label{eq:rao-value-stability}
|V_A-V_{\widetilde A}|
\le\|A-\widetilde A\|_{\max},
\qquad
|G_A-G_{\widetilde A}|
\le2\|A-\widetilde A\|_{\max},
\end{equation}
where \(\|B\|_{\max}=\max_{s,t}|b_{st}|\).
These bounds require neither conditional negative semidefiniteness nor
uniqueness and have constants independent of \(S\).
\item If, moreover, for some \(\kappa>0\),
\( -v^\top Av\ge\kappa\|v\|_2^2\) for all \(v\in\mathcal T\),
then the maximizing composition \(p_A^*\) is unique and, for any
maximizer \(p_{\widetilde A}^*\) over the same set,
\begin{equation}
\label{eq:rao-composition-stability}
\|p_A^*-p_{\widetilde A}^*\|_2
\le\frac{\|A-\widetilde A\|_{\mathrm{op}}}{\kappa}.
\end{equation}
\end{enumerate}
\end{theorem}

\begin{proof}
The feasible set is compact and the quadratic objective is continuous,
giving existence.  Feasibility of \(p^{\mathrm{obs}}\) gives \(G_A\ge0\).
For \(p,q\in\mathcal C\) and \(0<t<1\),
\[
Q_A((1-t)p+tq)
=(1-t)Q_A(p)+tQ_A(q)
-t(1-t)(p-q)^\top A(p-q).
\]
Since \(p-q\in\mathcal T\), conditional negative semidefiniteness on
\(\mathcal T\) implies concavity, and strict negativity implies strict
concavity and uniqueness.

Let \(\varepsilon=\|A-\widetilde A\|_{\max}\).  For every
\(p\in\Delta_{S-1}\),
\[
|p^\top(A-\widetilde A)p|
\le\varepsilon\sum_{s,t}p_sp_t=\varepsilon.
\]
Taking maxima in both directions yields the first inequality in
\eqref{eq:rao-value-stability}; applying the same bound to the observed
composition yields the second.

For the final assertion, write \(d=p_A^*-p_{\widetilde A}^*\in\mathcal T\).
The first-order conditions for the two global maximizers on the convex
set \(\mathcal C\) imply
\[
d^\top(Ap_A^*-\widetilde A p_{\widetilde A}^*)\ge0.
\]
Consequently,
\[
\kappa\|d\|_2^2
\le d^\top(A-\widetilde A)p_{\widetilde A}^*
\le\|d\|_2\|A-\widetilde A\|_{\mathrm{op}},
\]
where \(\|p_{\widetilde A}^*\|_2\le1\).  This gives
\eqref{eq:rao-composition-stability}.
\end{proof}

For squared Euclidean dissimilarities,
\begin{equation}
\label{eq:rao-squared-euclidean-curvature}
-v^\top Av=2\left\|\sum_s v_sz_s\right\|^2
\qquad(\mathbf1^\top v=0).
\end{equation}
Thus the strict-curvature hypothesis can fail when the trait dimension is
small relative to the number of species.  If \(\Lambda\) is not constant,
\(\dim\mathcal T=S-2\); a sufficient condition based on injectivity of
\(v\mapsto\sum_s v_sz_s\) cannot hold when the trait dimension is below
\(S-2\).  This does not rule out a unique boundary maximizer, but it
explains why uniqueness cannot generally be inferred from concavity alone.
If the constraint set also fixes the trait mean
\(\sum_s p_sz_s=m\), then
\(Q_A(p)=2\sum_s p_s\|z_s\|^2-2\|m\|^2\) is affine on that set,
another reason why a strictly concave Rao objective cannot be assumed.
When \(A\) is conditionally negative semidefinite on the entire simplex,
the quadratic form also induces the nonnegative Bregman-type discrepancy
\[
\mathcal D_{-Q_A}(p\|r)
=-(p-r)^\top A(p-r)
=2\left\|\sum_s(p_s-r_s)z_s\right\|^2
\]
in the squared Euclidean case.  It may vanish for distinct compositions
with identical trait centroids; it is not the Fisher--Rao distance.

\subsection{Entropic selection of Rao-optimal compositions}
\label{subsec:rao-entropic-selection}

When a unique benchmark composition is needed, the reference distribution
used in Section~\ref{sec:projection} provides a principled selection rule.
For \(r\) with \(r_s>0\), let
\[
\mathcal D_{\mathrm{KL}}(p\|r)=\sum_{s=1}^S p_s\log(p_s/r_s),
\]
with \(0\log0=0\), and define the regularized potential
\begin{equation}
\label{eq:rao-regularized}
p_{R,\tau}^{\mathrm{pot}}(x)
=\underset{p(x)\in\mathcal C_R(x)}{\mathrm{argmax}}
\left\{
Q(p(x))-\tau\mathcal D_{\mathrm{KL}}(p(x)\|r(x))
\right\},\qquad\tau>0.
\end{equation}
For a conditionally negative semidefinite \(A\), this objective is strictly
concave on the affine capacity set.  If \(A\) is not conditionally negative
semidefinite, an arbitrary positive \(\tau\) does not by itself guarantee
concavity or uniqueness.

\begin{theorem}[Entropic selection and stability of Rao-optimal compositions]
\label{thm:rao-entropic-selection}
Let \(\mathcal C\subseteq\Delta_{S-1}\) be nonempty, compact, and convex.
Assume that \(A=A^\top\) is conditionally negative semidefinite on
\(\mathcal C-\mathcal C\), and \(r_s>0\) for all \(s\).  Define
\[
V_A=\max_{p\in\mathcal C}Q_A(p),\qquad
\mathcal M_A=\arg\max_{p\in\mathcal C}Q_A(p),
\]
and let \(p^\dagger\) minimize \(\mathcal D_{\mathrm{KL}}(p\|r)\)
over \(\mathcal M_A\).  Then \(p^\dagger\) is unique.  For every
\(\tau>0\), the regularized maximizer \(p_\tau(A)\) is unique and
\begin{equation}
\label{eq:rao-regularization-bias}
0\le V_A-Q_A(p_\tau(A))
\le\tau\mathcal D_{\mathrm{KL}}(p^\dagger\|r)
\le\tau\log\frac{1}{\min_s r_s}.
\end{equation}
Moreover,
\begin{equation}
\label{eq:rao-entropic-limit}
p_\tau(A)\longrightarrow p^\dagger
\qquad\text{as }\tau\downarrow0.
\end{equation}
For a fixed \(\mathcal C,r,\tau\) and a second symmetric matrix
\(\widetilde A\) satisfying the same concavity condition,
\begin{equation}
\label{eq:rao-regularized-lipschitz}
\|p_\tau(A)-p_\tau(\widetilde A)\|_2
\le\frac{2}{\tau}\|A-\widetilde A\|_{\mathrm{op}}.
\end{equation}
\end{theorem}

\begin{proof}
Concavity of \(Q_A\) implies that \(\mathcal M_A\) is convex and compact.
Strict convexity of the KL divergence gives the unique minimizer
\(p^\dagger\).  The regularized objective is strictly concave, so
\(p_\tau(A)\) is unique.  Its optimality gives
\[
Q_A(p_\tau(A))-\tau\mathcal D_{\mathrm{KL}}(p_\tau(A)\|r)
\ge V_A-\tau\mathcal D_{\mathrm{KL}}(p^\dagger\|r).
\]
Since \(Q_A(p_\tau(A))\le V_A\) and KL divergence is nonnegative,
the first two inequalities in \eqref{eq:rao-regularization-bias} follow.
The final inequality follows from
\(\mathcal D_{\mathrm{KL}}(p\|r)\le-\log(\min_s r_s)\)
for every composition \(p\).

The bound shows that \(Q_A(p_\tau(A))\to V_A\).  Every accumulation
point as \(\tau\downarrow0\) therefore lies in \(\mathcal M_A\).
Comparing the regularized objective with any \(q\in\mathcal M_A\)
also yields \(\mathcal D_{\mathrm{KL}}(p_\tau(A)\|r)
\le\mathcal D_{\mathrm{KL}}(q\|r)\).  By compactness and continuity,
every accumulation point must equal \(p^\dagger\), proving
\eqref{eq:rao-entropic-limit}.

Finally, the negative KL term makes the regularized objective
\(\tau\)-strongly concave in the Euclidean norm on the simplex:
its negative Hessian contributes
\(\tau\sum_s v_s^2/p_s\ge\tau\|v\|_2^2\) in tangent directions.
Apply the strong-concavity optimality inequality at the maximizers for
\(A\) and \(\widetilde A\), then add the two inequalities.  The KL terms
cancel and give, with \(d=p_\tau(A)-p_\tau(\widetilde A)\),
\[
\tau\|d\|_2^2
\le d^\top(A-\widetilde A)
\{p_\tau(A)+p_\tau(\widetilde A)\}
\le2\|d\|_2\|A-\widetilde A\|_{\mathrm{op}}.
\]
This proves \eqref{eq:rao-regularized-lipschitz}; the same argument
extends to boundary-feasible sets by continuity.
\end{proof}

The theorem separates two uses of regularization.  The limiting
\(p^\dagger\) selects a unique composition \emph{among} unregularized Rao
maximizers without changing \(V_A\).  At finite \(\tau\), the regularized
composition may have smaller Rao diversity than an observed feasible
composition: the difference
\(Q_A(p_\tau(A))-Q_A(p^{\mathrm{obs}})\) need not be nonnegative.
The nonnegative Rao gap therefore remains
\(V_A-Q_A(p^{\mathrm{obs}})\), whereas a gap defined from the
regularized objective would be a different index.

The stability bounds above hold with a \emph{fixed} feasible set.
They cannot be transferred automatically to perturbations of ecological
capacity.  For example, with two species, \(a_{12}=a_{21}=1\),
\(p^{\mathrm{obs}}=(1/4,3/4)\), and
\(\Lambda_\varepsilon=(0,\varepsilon)\), take
\(c_\varepsilon=3\varepsilon/4\).  For every \(\varepsilon>0\),
the equality constraint admits only \(p^{\mathrm{obs}}\), so the Rao gap
is zero.  At \(\varepsilon=0\) the constraint becomes vacuous and the
gap equals \(1/2-3/8=1/8\).  Thus even the optimal value may be
discontinuous when the constraint rank changes.  Positive feasible
compositions and nondegenerate capacity contrasts provide natural
sufficient conditions for local sensitivity results in \(\Lambda\), but
these require a separate analysis of the constraint system.

\subsection{Stability of Rao potential diversity under trait perturbations}
\label{subsec:rao-stability-experiment}

We provide a controlled experiment illustrating the distinction between
stability of the maximal Rao diversity and stability of a maximizing
composition.  The experiment uses four species with one-dimensional traits
and a fixed ecological capacity constraint.  Let
\[
z(\varepsilon)=(0,\,1+\varepsilon,\,1-\varepsilon,\,1/2),
\qquad
a_{st}(\varepsilon)=|z_s(\varepsilon)-z_t(\varepsilon)|^2.
\]
The dissimilarity matrix \(A(\varepsilon)\) is conditionally negative
semidefinite for every \(\varepsilon\).  The capacity vector, budget, and
observed composition are fixed at
\[
\Lambda=(0,1,1,1/2),\quad
c=0.55,\quad
p^{\mathrm{obs}}=(0.30,0.20,0.20,0.30).
\]
Thus \(\Lambda^\top p^{\mathrm{obs}}=c\), and all experiments use the same
nonempty feasible set \(\mathcal C=\{p\in\Delta_3:\Lambda^\top p=c\}\).

For \(0<|\varepsilon|\le0.05\), direct maximization gives
\[
p^*(\varepsilon)=
\begin{cases}
(0.45,0.55,0,0),&\varepsilon>0,\\
(0.45,0,0.55,0),&\varepsilon<0.
\end{cases}
\]
At \(\varepsilon=0\), the Rao maximizer is not unique: every composition
\((0.45,t,0.55-t,0)\) with \(0\le t\le0.55\) is optimal.  The
KL-minimal solution relative to the observed composition is
\(p^\dagger=(0.45,0.275,0.275,0)\).  Nevertheless, the optimal value is
continuous and has the explicit expression
\begin{equation}
\label{eq:rao-stability-exact-value}
V(\varepsilon)=0.495(1+|\varepsilon|)^2.
\end{equation}
The observed diversity and the Rao gap are
\[
Q_{A(\varepsilon)}(p^{\mathrm{obs}})
=0.345+0.8\varepsilon^2,
\qquad
G(\varepsilon)=0.15+0.99|\varepsilon|-0.305\varepsilon^2.
\]
In particular, \(V(0)=0.495\) and \(G(0)=0.150\), whereas
\(V(\pm0.01)=0.5049495\) and \(G(\pm0.01)=0.1598695\).  The two
maximizing compositions at \(\pm0.01\) are separated by an \(L^1\)
distance of \(1.10\), despite having identical maximal Rao values.
For \(|\varepsilon|\le0.05\),
\[
\|A(\varepsilon)-A(0)\|_{\max}
=2|\varepsilon|+\varepsilon^2,
\]
so the value and gap satisfy the perturbation bounds of
Theorem~\ref{thm:rao-wellposed-stable} for every value in the displayed
range (Figure~\ref{fig:rao-stability-main}, panels A--C).

We then solve the KL-regularized problem with reference
\(r=p^{\mathrm{obs}}\) and \(\tau\in\{0.005,0.020,0.050\}\).
For the two perturbations \(\varepsilon=\pm0.001\), the \(L^1\) distance
between the corresponding regularized compositions decreases from \(1.10\)
without regularization to \(0.3796\), \(0.0987\), and \(0.0396\),
respectively.  All computed regularization biases
\(V(\varepsilon)-Q_{A(\varepsilon)}(p_\tau)\) satisfy the
upper bound in Theorem~\ref{thm:rao-entropic-selection}.
The regularized composition varies continuously with \(\varepsilon\)
for every fixed \(\tau>0\) (Figure~\ref{fig:rao-stability-main}, panel D).

\begin{figure}[htbp]
\centering
\includegraphics[width=\textwidth]{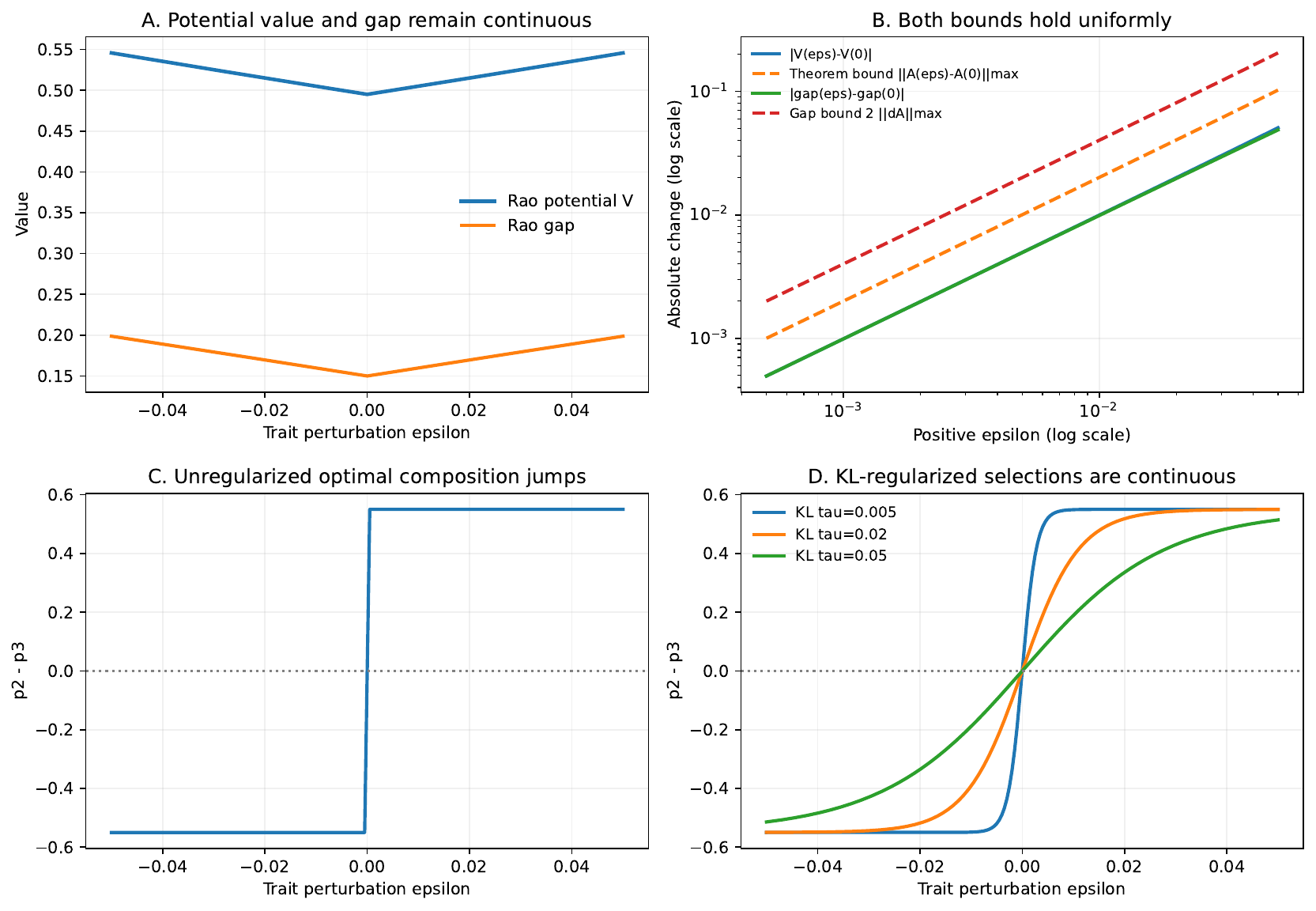}
\caption{Perturbation of a four-species Rao maximum with a fixed capacity
constraint.  (A) The optimal Rao value and gap remain continuous.
(B) Their changes remain below the theoretical perturbation bounds.
(C) A nonunique maximum at \(\varepsilon=0\) produces a discontinuous
unregularized optimal composition.  (D) KL regularization smooths the
selection for fixed positive \(\tau\).}
\label{fig:rao-stability-main}
\end{figure}

To check the perturbation bounds beyond this deterministic path, we
performed 300 independent perturbations at each of six noise levels
\(\sigma\in\{10^{-4},5\times10^{-4},10^{-3},5\times10^{-3},
10^{-2},5\times10^{-2}\}\), adding independent
\(N(0,\sigma^2)\) noise to all four traits and reconstructing the squared
distance matrix. All 1,800 replications satisfied both the
optimal-value and Rao-gap bounds. Table~\ref{tab:rao-stability-mc}
summarizes the Monte Carlo results. The median absolute change in the
optimal value increases with the noise scale, whereas the median
\(L^1\) deviation of the selected maximizer from the KL-minimal
unperturbed composition remains \(0.55\) at every noise level.

\begin{table}[htbp]
\centering
\small
\caption{Monte Carlo stability of the Rao maximum and its maximizing
composition under independent Gaussian perturbations of the four traits
(300 replications per noise level). The value error is measured relative
to \(V(A_0)=0.495\), and the composition error relative to the
KL-minimal unperturbed maximizer \(p^\dagger\).}
\label{tab:rao-stability-mc}
\medskip
\begin{tabular}{@{}ccc@{}}
\toprule
Noise SD \(\sigma\)
& Median \(\lvert V(A)-V(A_0)\rvert\)
& Median \(\lVert p^*-p^\dagger\rVert_1\) \\
\midrule
\(10^{-4}\)        & \(9.35\times10^{-5}\) & \(0.550\) \\
\(5\times10^{-4}\) & \(4.35\times10^{-4}\) & \(0.550\) \\
\(10^{-3}\)        & \(9.26\times10^{-4}\) & \(0.550\) \\
\(5\times10^{-3}\) & \(4.43\times10^{-3}\) & \(0.550\) \\
\(10^{-2}\)        & \(9.28\times10^{-3}\) & \(0.550\) \\
\(5\times10^{-2}\) & \(4.88\times10^{-2}\) & \(0.550\) \\
\bottomrule
\end{tabular}
\end{table}

The largest observed ratios of the value and gap changes to their
respective perturbation bounds were \(0.4950\) and \(0.2539\),
respectively (rounded upward). At the smallest noise level,
152 of the 300 selected compositions favored species 2 and the
remaining 148 favored species 3. Thus the optimal value remains
stable even when the optimizer fluctuates between two distinct modes.

These experiments deliberately probe a degenerate trait configuration
in which the optimizer is nonunique at the reference model.  They verify
the predicted distinction between a stable benchmark value and an unstable
choice of maximizing composition; they do not imply that discontinuity is
typical when a strict-curvature condition holds.  Moreover, the capacity
set was fixed throughout these perturbation experiments, as required by
the stability theorem.  The rank-degeneracy counterexample given above also shows that changing
capacity constraints can destabilize the benchmark value itself.

\subsection{Interpretation}

The Rao gap is not a species-richness gap.  It can be small even when some
species are absent, provided the remaining community spans the same
dissimilarity space.  Conversely, it can be large even when Hill diversity is
high, if the community consists of many functionally similar species.  Thus the
pair
\[
(\Delta_q^H(x),\Delta^R(x))
\]
separates two forms of biodiversity deficit: abundance/evenness deficit and
trait-sensitive spread deficit.

This distinction is especially important for conservation.  A community may
look diverse in the Hill sense because it has many relatively even species, but
it may still be functionally narrow.  Rao-type potential diversity identifies
this missing dimension by comparing the realized community with a benchmark
that expands dissimilarity spread under the same local constraints.

\subsection{A schematic illustration}

We give a schematic numerical illustration to clarify the difference between
Hill-type and Rao-type diversity gaps.  The purpose of this illustration is not
to mimic a realistic ecological system or to validate an empirical estimation
procedure.  Rather, it shows that an abundance-based gap and a
dissimilarity-based gap can respond differently to the same change in community
composition.

The example uses \(S=6\) species divided into two functional groups,
\(\{1,2,3\}\) and \(\{4,5,6\}.\)
The dissimilarity matrix is chosen so that species within the same functional
group are close, whereas species belonging to different functional groups are
distant.  A smooth potential composition \(p^{\mathrm{pot}}(x)\) is first
specified on the unit square.  The observed composition
\(p^{\mathrm{obs}}(x)\) is then obtained by imposing two schematic degradation
mechanisms.  Region A represents functional filtering, where one functional
group is depleted.  Region B represents within-group homogenization, where
species-level evenness is lost within functional groups while the broad
functional groups remain present.

Figure~\ref{fig:schematic-hill-rao-gaps} gives a schematic illustration of the
distinction between Hill-type and Rao-type diversity gaps.  The upper row shows
the effective Shannon diversity \(D_1\) and its potential gap, whereas the lower
row shows Rao's quadratic entropy \(Q\) and its potential gap.  Region A
represents functional filtering, where one functional group is depleted and the
Rao-type gap is emphasized.  Region B represents within-group homogenization,
where species-level evenness is lost within functional groups and the Hill-type
gap is emphasized.  Thus the two gaps diagnose different forms of unrealized
diversity, even though they are computed from the same pair
\((p^{\mathrm{obs}},p^{\mathrm{pot}})\).

\begin{figure}[H]
\centering
\includegraphics[width=\textwidth]{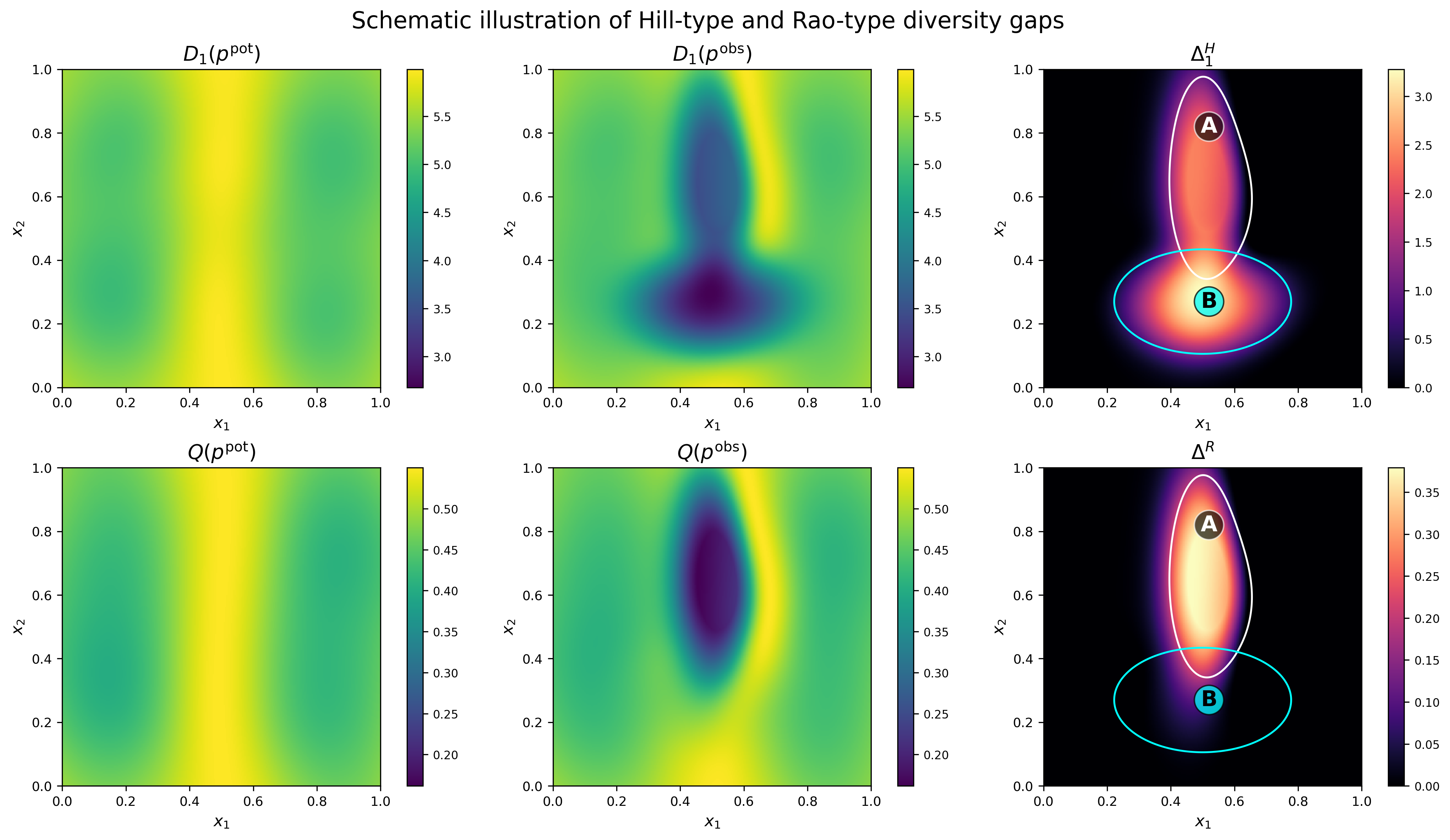}
\caption{
Schematic illustration of Hill-type and Rao-type potential diversity gaps.
The upper row displays \(D_1(p^{\mathrm{pot}})\),
\(D_1(p^{\mathrm{obs}})\), and \(\Delta^H_1\), while the lower row displays
\(Q(p^{\mathrm{pot}})\), \(Q(p^{\mathrm{obs}})\), and \(\Delta^R\).
}
\label{fig:schematic-hill-rao-gaps}
\end{figure}

\section{Hill-type potential diversity}
\label{sec:hill-potential}

We now turn to Hill numbers as a complementary, abundance-sensitive
benchmark to the Rao construction.  Fix \(q\ge 0\).
A Hill-type potential composition is defined by
\begin{equation*}
p_q^{\mathrm{pot}}(x)
\in
\underset{p(x)\in\mathcal C_q(x)}{\mathrm{argmax}}
D_q(p(x)),
\end{equation*}
where \(\mathcal C_q(x)\) is a constraint set.  For \(q>0\),
a natural choice is the escort capacity set
\begin{equation*}
\mathcal C_q(x)
=
\left\{
p(x)\in\Delta_{S-1}:
E^{(q)}[\Lambda(x);p(x)]
=
E^{(q)}[\Lambda(x);p^{\mathrm{obs}}(x)]
\right\}.
\end{equation*}
This choice makes the observed composition feasible and aligns the order of the
constraint with the order of the diversity functional.

For \(q=0\), the escort expression does not provide the same unambiguous
constraint at compositions with zero components.  We therefore specify
a feasible set separately, for example the ordinary linear capacity set
\[
\mathcal C_0(x)
=\left\{p\in\Delta_{S-1}:
\sum_s p_s\Lambda_s(x)
=\sum_s p_s^{\mathrm{obs}}(x)\Lambda_s(x)\right\}.
\]
This distinction keeps potential richness well defined without assigning
a value to \(0^0\) in an escort distribution.

The Hill-type diversity gap is
\begin{equation}
\label{eq:hill-gap}
\Delta_q^H(x)
=
D_q(p_q^{\mathrm{pot}}(x))
-
D_q(p^{\mathrm{obs}}(x)).
\end{equation}
For \(q=0\), the gap is a potential richness gap.  For \(q=1\), it is an
effective Shannon-diversity gap.  For \(q=2\), it is an inverse-Simpson gap and
therefore emphasizes dominance.

\subsection{Maximum entropy form}
\label{subsec:maxent}

For \(q=1\), the maximization of \(D_1\) is equivalent to the maximization of
Shannon entropy.  Under the linear capacity constraint
\[
\sum_{s=1}^S p_s \Lambda_s(x)=c(x),
\]
the Lagrange equations give the Gibbs form
\begin{equation}
\label{eq:gibbs-pot}
p_s^{\mathrm{pot}}(x)
=
\frac{\exp\{\theta(x)\Lambda_s(x)\}}
{\sum_{k=1}^S \exp\{\theta(x)\Lambda_k(x)\}},
\end{equation}
where \(\theta(x)\) is chosen so that the capacity constraint is satisfied.
The sign and magnitude of \(\theta(x)\) determine whether the benchmark puts
more mass on high-capacity or low-capacity species.  The constraint budget
\(c(x)\) fixes this choice.

\subsection{\texorpdfstring{\(q\)}{q}-exponential form}

For \(0<q<1\), the Tsallis entropy representation leads to a
\(q\)-exponential form.  Under an escort capacity constraint,
\[
E^{(q)}[\Lambda(x);p]=c_q(x),
\]
a regular maximizing distribution has the form
\begin{equation}
\label{eq:qexp-pot}
p_s^{\mathrm{pot}}(x)
=
\frac{
\left[
1-(1-q)\beta(x)\{\Lambda_s(x)-c_q(x)\}
\right]_+^{1/(1-q)}
}
{
\sum_{k=1}^S
\left[
1-(1-q)\beta(x)\{\Lambda_k(x)-c_q(x)\}
\right]_+^{1/(1-q)}
},
\end{equation}
where \(\beta(x)\) is a Lagrange multiplier and
\([u]_+=\max(u,0)\) permits a boundary optimum.
For \(q>1\), the interior solution instead reads
\[
p_s^{\mathrm{pot}}(x)
=\frac{
\{1+(q-1)\beta(x)[\Lambda_s(x)-c_q(x)]\}^{-1/(q-1)}
}{
\sum_{k=1}^S
\{1+(q-1)\beta(x)[\Lambda_k(x)-c_q(x)]\}^{-1/(q-1)}
},
\]
provided every expression inside braces is strictly positive.
The positive-part operation must not be used with the negative exponent
\(-1/(q-1)\).  When the capacity budget forces some components to vanish,
the optimization is carried out on the corresponding feasible support.

The global character of these solutions is clarified by the escort change
of variables \(w_s=p_s^q/\sum_k p_k^q\).  Its inverse and the Hill number are
\[
p_s=\frac{w_s^{1/q}}{\sum_k w_k^{1/q}},
\qquad
D_q(p)=\left(\sum_{s=1}^S w_s^{1/q}\right)^{q/(q-1)}.
\]
The escort constraint becomes the affine equation
\(\sum_s w_s\Lambda_s(x)=c_q(x)\).  Thus, for \(q>1\),
maximizing \(D_q\) amounts to maximizing the strictly concave function
\(\sum_s w_s^{1/q}\); for \(0<q<1\), it amounts to minimizing
the strictly convex function \(\sum_s w_s^{1/q}\).  Whenever the
feasible set is nonempty, the global maximizing composition is therefore
unique.  The \(q\)-exponential expressions above describe its regular
Lagrange-multiplier forms; see \citet{eguchi2025information} for related
derivations.

Equation~\eqref{eq:qexp-pot} and its \(q>1\) counterpart clarify
why the potential distribution is not generally uniform.  The benchmark preserves a local capacity budget and then
maximizes diversity within that constraint.  Thus the potential composition is
a constrained evenness state, not an unconstrained evenness state.

\myparagraph{Cross-Hill representation under an ordinary linear constraint.}
There is a useful divergence representation of the Hill gap when the
capacity constraint is linear in \(p\), rather than in its escort
distribution.  Fix \(0<q<1\), and let
\[
\mathcal C^{\mathrm{lin}}(x)
=\left\{p\in\Delta_{S-1}:
\sum_s p_s\Lambda_s(x)=c(x)\right\}.
\]
For strictly positive \(u\), define the cross-Hill number by
\[
D_q^{\mathrm{cross}}(\pi,u)
=\left(\sum_s u_s^q\right)
 \left(\sum_s\pi_su_s^{q-1}\right)^{q/(1-q)}.
\]
As shown by \citet[Proposition~5]{eguchi2025information}, if
\(u=p_{q,\mathrm{lin}}^{\mathrm{pot}}(x)\) is the regular interior
maximizer of \(D_q\) over \(\mathcal C^{\mathrm{lin}}(x)\), then for
every \(\pi\in\mathcal C^{\mathrm{lin}}(x)\),
\begin{equation}
\label{eq:cross-hill-linear-gap}
D_q(u)-D_q(\pi)
=D_q^{\mathrm{cross}}(\pi,u)-D_q(\pi)\ge0.
\end{equation}
Indeed, the multiplier equation gives
\(u_s^{q-1}=\alpha+\beta\Lambda_s(x)\), so that
\(\sum_s\pi_su_s^{q-1}=\sum_su_s^q\) whenever \(\pi\) and \(u\)
satisfy the same linear constraint.  Thus, in this setting, the
effective-number gap is exactly a cross-Hill divergence.

\begin{remark}[Linear and escort capacity constraints]
\label{rem:hill-linear-vs-escort}
The identity \eqref{eq:cross-hill-linear-gap} concerns the ordinary
linear capacity set and must not be transferred unchanged to the
escort set \(\mathcal C_q(x)\) used above.  Under the escort constraint,
the regular multiplier equation is affine in \(p_s^{1-q}\), rather
than in \(p_s^{q-1}\).  The cross-Hill identity therefore need not
hold.  For an interior optimum, however, the escort coordinates
\(w_s=p_s^q/\sum_kp_k^q\) make the constraint affine, and the gap
can be expressed through a Bregman divergence generated by
\(\sum_s w_s^{1/q}\), followed by the monotone transformation
\(D_q(p)=(\sum_s w_s^{1/q})^{q/(q-1)}\).
At \(q=1\), the two capacity constraints coincide, and an interior
Shannon maximizer \(u\) satisfies
\[
D_1(u)-D_1(\pi)
=D_1(u)\{1-\exp[-\operatorname{KL}(\pi\|u)]\}
\quad(\pi\in\mathcal C_1(x)).
\]
\end{remark}

\subsection{Numerical computation of the benchmark}
\label{subsec:hill-computation}

The \(q\)-exponential forms give a practical computation for \(q>0\),
\(q\ne1\).  At a fixed location \(x\), define
\[
B_s(\beta)
=
\begin{cases}
\left[1-(1-q)\beta\{\Lambda_s(x)-c_q(x)\}\right]_+^{1/(1-q)},
&0<q<1,\\
\left\{1+(q-1)\beta\{\Lambda_s(x)-c_q(x)\}\right\}^{-1/(q-1)},
&q>1,
\end{cases}
\]
where the \(q>1\) expression is evaluated only on the domain where
its base is strictly positive.
The corresponding normalized composition is
\[
p_s(\beta)
=
\frac{B_s(\beta)}{\sum_{k=1}^S B_k(\beta)}.
\]
The multiplier \(\beta\) is chosen so that the escort constraint is satisfied:
\begin{equation*}
F_x(\beta)
=
E^{(q)}[\Lambda(x);p(\beta)]-c_q(x)
=
0.
\end{equation*}
Thus the computation of \(p_q^{\mathrm{pot}}(x)\) is reduced to a scalar
root-finding problem on the appropriate multiplier domain when
the regular Lagrange-multiplier form applies.  A safeguarded Newton method or
bisection method is often sufficient, because the species dimension enters only
through the evaluation of \(B_s(\beta)\) and the normalized vector
\(p(\beta)\).  For \(0<q<1\), the positive-part operator can
identify a boundary support. For \(q>1\), a support forced by an endpoint
capacity budget must be treated separately. Alternatively, optimizing the
strictly convex or strictly concave objective in escort coordinates \(w\)
provides a globally valid formulation, including boundary cases.  In spatial applications,
this root-finding step is repeated independently over grid cells after the
capacity vector \(\Lambda(x)\) has been computed.

For \(q=1\), the same computation uses the Gibbs form
\eqref{eq:gibbs-pot}.  In that case one solves for \(\theta(x)\) in
\[
\sum_{s=1}^S p_s(\theta)\Lambda_s(x)-c(x)=0,
\qquad
p_s(\theta)
=
\frac{\exp\{\theta \Lambda_s(x)\}}
{\sum_{k=1}^S\exp\{\theta \Lambda_k(x)\}}.
\]
This is the usual one-dimensional maximum-entropy calibration problem.

\subsection{Hierarchical decomposition in the Shannon case}
\label{subsec:shannon-decomposition}

For \(q=1\), Shannon entropy gives an additional diagnostic decomposition.
Suppose that the species set is partitioned into disjoint groups
\(\mathcal G_1,\ldots,\mathcal G_K\), representing taxonomic, functional, or
phylogenetic categories.  For a composition \(p\), define
\[
P_k
=
\sum_{s\in\mathcal G_k}p_s,
\qquad
r_s^{(k)}
=
\frac{p_s}{P_k},
\quad s\in\mathcal G_k,
\]
whenever \(P_k>0\).  The entropy chain rule gives
\begin{equation*}
H(p)
=
H(P)+\sum_{k=1}^K P_k H(r^{(k)}).
\end{equation*}
Consequently, the Shannon entropy gap
\[
\Delta_1^{\mathrm{ent}}(x)
=
H(p_1^{\mathrm{pot}}(x))-H(p^{\mathrm{obs}}(x))
\]
decomposes as
\begin{align*}
\Delta_1^{\mathrm{ent}}(x)
&=
\left\{
H(P^{\mathrm{pot}}(x))-H(P^{\mathrm{obs}}(x))
\right\} \notag\\
&\quad+
\sum_{k=1}^K
\left\{
P_k^{\mathrm{pot}}(x)H(r^{(k),\mathrm{pot}}(x))
-
P_k^{\mathrm{obs}}(x)H(r^{(k),\mathrm{obs}}(x))
\right\}.
\end{align*}
The first term measures an unrealized between-group component, while the second
term measures unrealized within-group evenness.  This decomposition is useful
when a large gap may arise either from the loss of an entire guild or from the
hyper-dominance of a few species within otherwise feasible groups.

Because the Hill gap in \eqref{eq:hill-gap} uses the effective number
\(D_1=\exp\{H\}\), the additive formula above is most naturally applied on the
entropy scale.  The effective-number gap can then be reported alongside the
between- and within-group entropy components.

\begin{proposition}[Nonnegativity of the Hill gap]
\label{prop:hill-gap-nonnegative}
Assume that \(\mathcal C_q(x)\) is nonempty and that
\(p^{\mathrm{obs}}(x)\in\mathcal C_q(x)\).  If
\(p_q^{\mathrm{pot}}(x)\) is a maximizer of \(D_q\) over
\(\mathcal C_q(x)\), then
\[
\Delta_q^H(x)\ge 0.
\]
\end{proposition}

\begin{proof}
Since \(p^{\mathrm{obs}}(x)\in\mathcal C_q(x)\) and
\(p_q^{\mathrm{pot}}(x)\) maximizes \(D_q\) on \(\mathcal C_q(x)\),
\[
D_q(p_q^{\mathrm{pot}}(x))
\ge
D_q(p^{\mathrm{obs}}(x)).
\]
This is exactly \(\Delta_q^H(x)\ge 0\).
\end{proof}

The proposition is mathematically simple but conceptually important.  The gap
is not a residual from a fitted model.  It is a variational deficit relative to
an explicitly defined feasible benchmark.

\section{Dynamic extension}
\label{sec:dynamic}

The preceding sections are static.  In ecological applications, however, both
observed compositions and potential constraints vary over time.  Let
\[
p_t^{\mathrm{obs}}(x),\qquad
\Lambda_t(x),\qquad
A_t=(a_{st}(t))
\]
denote the observed composition, capacity vector, and dissimilarity matrix at
time \(t\).  The potential composition becomes
\[
p_t^{\mathrm{pot}}(x)
\in
\underset{p(x)\in\mathcal C(x)}{\mathrm{argmax}}
 D _t(p(x)),
\]
and the gap field is
\[
\Delta_{ D }(x,t)
=
 D _t(p_t^{\mathrm{pot}}(x))
-
 D _t(p_t^{\mathrm{obs}}(x)).
\]

This time-dependent formulation is a natural starting point for climate-change
applications.  Warming may change the capacity vector \(\Lambda_t(x)\), for example
by shifting the climatic suitability of species across latitude or elevation.
Migration and dispersal constraints may prevent the observed composition
\(p_t^{\mathrm{obs}}(x)\) from tracking the moving potential composition
\(p_t^{\mathrm{pot}}(x)\).  The resulting gap can be interpreted as a temporal
mismatch between ecological potential and realized community structure.

One may also introduce a dynamical system on the simplex.  A general form is
\begin{equation*}
\frac{d}{dt}p_t(x)
=
F_t(p_t(x),x),
\end{equation*}
where \(F_t\) represents colonization, extinction, competition, dispersal, and
environmental forcing.  
In an information-geometric formulation, one may
define a potential function
\[
\Psi_t(p,x)
=
D_t(p_t^{\mathrm{pot}}(x))-D_t(p)
\]
and consider a gradient-flow-type relaxation
\[
\frac{d}{dt}p_t(x)
=
-\operatorname{grad}_{g_t}\Psi_t(p_t(x),x),
\]
where \(g_t\) is a metric on the simplex.  Different choices of \(g_t\) lead
to different ecological dynamics.  The Fisher metric gives a
replicator-type geometry, while transport-type metrics would incorporate
spatial movement.

This dynamic viewpoint also suggests a robust-inference extension.  In
large-scale biodiversity data, the estimated composition field \(p_t(x)\) and
the capacity field \(\Lambda_t(x)\) may be affected by outlying observations,
uneven sampling effort, or abrupt local disturbances.  Density-powered Stein
operators provide one possible information-geometric tool for constructing
robust estimating equations and robust gradient-type dynamics for such
density fields \citep{eguchi2025robust}.  In the present chapter we do not
develop this direction, but it points to a possible connection between
potential diversity gaps and robust geometric inference for evolving ecological
systems.
 
The dynamic extension suggests three future directions.  First, one can study
whether observed communities lag behind climate-driven potential distributions.
Second, one can decompose gap changes into capacity shifts and composition
shifts.  Third, one can formulate restoration or assisted-migration strategies
as interventions that reduce the dynamic gap.

\section{Concluding remarks}
\label{sec:conclusion}

This chapter proposed an information-geometric framework for potential
biodiversity and diversity gaps.  The main message is that biodiversity
measurement should not be limited to descriptive summaries of realized
communities.  For conservation and ecological planning, it is equally important
to define a potential benchmark and to measure the gap from that benchmark.

The framework places Rao's quadratic entropy at the center of potential
biodiversity assessment, with Hill numbers as a complementary account of
abundance and evenness.  Rao-type gaps measure deficits in trait,
phylogenetic, or ecological dissimilarity spread.  The optimal Rao diversity
is well defined even when its maximizing composition is nonunique, and the
gap is uniformly stable under elementwise perturbations of the dissimilarity
matrix when the feasible set is held fixed.  KL regularization provides a
unique composition with a controlled loss in the Rao objective, while
ecological capacity constraints are a separate source of uncertainty.  Both
Rao and Hill gaps compare an observed composition with a potential
composition defined by a constrained variational principle.  The same
comparison also gives a continuous version of dark diversity: not merely a list
of absent species, but a benchmark-relative abundance and evenness deficit on
the simplex.  The Shannon case further permits a hierarchical decomposition of
the gap into between-group and within-group components.

The chapter deliberately avoided empirical analysis.  A full empirical study
requires careful modeling of observation processes, spatial sampling bias,
temporal variation, and the construction of species dissimilarity matrices.
Large-scale citizen-science data, such as eBird, and trait databases, such as
global avian trait resources, provide promising material for such work.  A
particularly important direction is the extension to time-dependent potential
diversity, where climate warming, species migration, and habitat change induce
dynamic shifts in the diversity gap.

The proposed framework should therefore be viewed as a theoretical foundation.
Its role is to clarify what should be compared, under what constraints, and
with which diversity functional.  Once these choices are made explicit,
empirical biodiversity mapping can move from descriptive maps toward
benchmark-relative ecological diagnosis.

\bibliographystyle{plainnat-eguchi}
\bibliography{refs}

\end{document}